\documentclass[11p,reqno]{amsart}
\usepackage[T1]{fontenc}
\usepackage{graphicx}
\usepackage{amssymb,amsthm,amsmath,mathrsfs,bm,braket,marginnote}
\usepackage{enumerate}
\usepackage{appendix}
\usepackage[colorlinks=true, pdfstartview=FitV, linkcolor=blue, citecolor=blue, urlcolor=blue]{hyperref}
\usepackage{multirow}

\usepackage{pgf}
\usepackage{pgfplots}
\usepackage{tikz}
\usetikzlibrary{arrows,calc}
\usepackage{verbatim}
\usetikzlibrary{decorations.pathreplacing,decorations.pathmorphing}
\usepackage[numbers,sort&compress]{natbib}
\usepackage{dsfont}

\numberwithin{equation}{section}
\newtheorem{theorem}{Theorem}[section]

\newtheorem{coro}[theorem]{Corollary}

\newtheorem{proposition}[theorem]{Proposition}

\numberwithin{equation}{section}

\begin{document}

\title[A direct method in noncommutative integrable systems]{A direct method in noncommutative integrable systems}

\date{}

\dedicatory{}


\author{Shi-Hao Li}
\address{Department of Mathematics, Sichuan University, Chengdu, 610064, PR China}
\email{shihao.li@scu.edu.cn}

\author{Shou-Feng Shen}
\address{School of Mathematical Sciences, Zhejiang University of Technology, Hangzhou, 310023, Zhejiang, China}
\email{mathssf@zjut.edu.cn}

\author{Guo-Fu Yu}
\address{School of Mathematical Sciences, Shanghai Jiaotong University, People's Republic of China; School of Mathematical Sciences, CMA-Shanghai, Shanghai Jiao Tong University, Shanghai 200240}
\email{gfyu@sjtu.edu.cn}

\author{Jun-Yang Zhang}
\address{School of Mathematical Sciences, Zhejiang University of Technology, Hangzhou, 310023, Zhejiang, China}
\email{zhangjunyang142@gmail.com}

\begin{abstract}
We present a constructive framework for deriving noncommutative (NC) integrable equations directly from quasi-determinant solutions. Building upon the quasi-Wronskian structure, we extend the classical direct method to the NC setting, where standard determinant identities are replaced by algebraic relations intrinsic to quasi-determinants. By analyzing derivative identities satisfied by quasi-determinants, we recover the NC Kadomtsev-Petviashvili (ncKP) and Date-Jimbo-Kashiwara-Miwa (ncDJKM) equations by cancellations of nonlinear terms. Furthermore, by imposing flow constraints on the seed functions, we derive NC reductions such as the ncKdV and ncBoussinesq equations and obtain explicit matrix-valued soliton solutions. Our results highlight the quasi-determinant as a fundamental algebraic structure underpinning NC $\tau$-function structure.

\end{abstract}

\maketitle

\section{Introduction}

Direct method, or the so-called Hirota's bilinear method is one of the most effective way in solving nonlinear integrable systems. It has exclusively exhibited by Hirota that the bilinear equations, which is obtained by taking the dependent variables into the nonlinear equation, is equivalent to algebraic 
identities of determinants or Pfaffians \cite{hirota04}. This genius idea is then generalized by the Kyoto school, showing that the nonlinear integrable systems could be realized as some orbits in Grassmannian, and the determinantal and Pfaffian identities are satisfied by Pl\"ucker coordiates on the Grassmannian \cite{sato83}. The solution of a bilinear equation, often referred to as a $\tau$-function, encapsulates the essential algebraic structure of the system, and the integrable equation arises as an identity among its derivatives  \cite{hirota71,hirota04,sato83,date82}.
Interestingly, the above facts could be understood backwards. \emph{If one starts with a determinant or Pfaffian structure, then certain integrable nonlinear equations could be obtained with proper dispersion relations.} This idea has been widely applied in the theory of integrable system in constructing bilinear B\"acklund transformation \cite{nimmo84,nimmo83}, reduction theory in rogue waves \cite{ohta14,feng22}, among many other fields.

The purpose of this work is to extend this philosophy to the noncommutative (NC) setting, where  determinants are no longer applicable. The quasi-determinants, introduced by Gelfand and Retakh \cite{gelfand91,gelfand05}, serve as a suitable generalization for matrices over noncommutative division rings. We formulate a framework in which NC integrable equations, such as the noncommutative KP (ncKP) hierarchy and its reductions, could be directly derived from quasi-determinant solutions, notably the quasi-Wronskian or quasi-Grammian types. This can be viewed as a noncommutative generalization of Hirota's bilinear method. This approach builds on the work of Gilson and Nimmo \cite{gilson07}, who constructed explicit quasi-determinant solutions of the ncKP equation by using Darboux and binary Darboux transformations. In contrast to verifying quasi-determinants satisfy some known noncommutative equations, our focus is on deriving these equations from the algebraic structure of the solution. In particular, we show that the reduction theory in bilinear method can also be applied to the noncommutative case. We show the validity of our methods by getting the noncommutative KdV (ncKdV) and noncommutative Boussinesq (ncBoussinesq) equations, as well as their solutions. 

In physical literatures, NC integrable systems have received considerable attention due to their links with noncommutative geometry, deformation quantization, and gauge theory \cite{szabo03,castellani00,asakawa00,landi03}. Notable examples include noncommutative generalizations of the KP and KdV equations, where the underlying function algebra is deformed into a noncommutative associative algebra, often via the Moyal product \cite{hamanaka10,paniak01}. Geometrically, these systems arise from reductions of the noncommutative (anti-)self-dual Yang-Mills equations, supporting the noncommutative Ward conjecture \cite{hamanaka03}. Their hierarchies retain rich structure, including commuting flows and conserved quantities \cite{hamanaka05}, and their solution spaces share structural similarities with the commutative case via Birkhoff factorization \cite{sakakibara04}. Explicit multi-soliton solutions have been constructed using quasi-determinants and deformation techniques \cite{hamanaka07}. It should be mentioned that the quasi-determinants technique is not only useful in verification of solutions for NC equations, but applicable to construct NC B\"acklund transformation \cite{gilson08}.  The technique of quasi-determinants also makes tight connections with Somos sequence in combinatorics \cite{bobrova24}, noncommutative leapfrog map in projective geometry and networks \cite{wang25}, 
Wynn recurrence in numerical algorithms \cite{doliwa23}, and matrix-valued orthogonal polynomials \cite{li24,gilson25}. 

The structure and main results of this paper are organized as follows.
\begin{itemize}
    \item In Section 2, we review the ncKP hierarchy from the pseudo-differential Lax formalism. Reductions to the ncKdV and ncBoussinesq equations are obtained by constraining the operator structure.

    \item In Section 3, we begin with a quasi-determinant solution, specifically the quasi-Wronskian, and reconstruct the ncKP and noncommutative Date-Jimbo-Kashiwara-Miwa (ncDJKM) equations by analyzing the cancellation of nonlinear terms in higher-order derivatives. We also present matrix-valued soliton solutions of one- and two-soliton type.

    \item In Sections 4, we revisit the ncKdV and ncBoussinesq equations by imposing variable-independence constraints on the quasi-determinant. We show that the quasi-determinant remains a valid solution under such reductions.
\end{itemize}

This work establishes a framework for deriving NC integrable equations directly from quasi-determinant solutions. It generalizes the classical direct method to the noncommutative setting and reveals the quasi-determinant as a fundamental object from which integrable structures can be systematically generated.

\section{NC Integrable Equations}
\noindent In this section, we derive several NC integrable equations based on the ncKP hierarchy. First, let us introduce a pseudo-differential operator $L$ which include negative powers of differential operators
\begin{equation}
    L = \partial_{t_1} + u_2 \partial_{t_1}^{-1} + u_3 \partial_{t_1}^{-2} + u_4 \partial_{t_1}^{-3} + \cdots,
\end{equation}
where $u_k$ ($k = 2,3, \ldots$) are some associative variables but not necessarily commutative, depending on parameters $t_q$ ($q = 1, 2, \ldots$). The action of a differential operator $\partial_{t_1}^n$ on a multiplicity operator $f$ is formally defined as the following generalized Leibniz rule
\begin{equation}
    \partial_{t_1}^n \cdot f := \sum_{i \geq 0} \binom{n}{i} (\partial_{t_1}^i f) \partial_{t_1}^{n-i},
\end{equation}
where the binomial coefficient is given by
\begin{equation}
    \binom{n}{i} := \frac{n (n-1) \cdots (n-i+1)}{i (i-1) \cdots 1}.
\end{equation}

\subsection{ncKP Hierarchy}
Similar to commutative integrable equations, 
the ncKP hierarchy could be defined by Lax representations (c.f. \cite{hamanaka10})
\begin{equation} \label{lax_equation}
    \partial_{t_m} L = \left[B_m, L\right], \quad \mathrm{with} \quad B_m = (L^m)_{\geq 0}, \quad m = 1, 2, \ldots,
\end{equation}
where $(L^m)_{\geq 0}$ represents the polynomial part of $\partial_{t_1}$ in $L^m$, i.e. $B_m$ is a differential operator of order $m$, and $\left[B_m, L\right] := B_m L - L B_m$ is the usual commutator of operators. The first four such
differential operators are
\begin{equation}
    \begin{split}
        &B_1 = \partial_{t_1}, \\
        &B_2 = \partial_{t_1}^2 + 2 u_2, \\
        &B_3 = \partial_{t_1}^3 + 3 u_2 \partial_{t_1} + 3 (\partial_{t_1} u_2 + u_3), \\
        &B_4 = \partial_{t_1}^4 + 4 u_2 \partial_{t_1}^2 + 2 (3\partial_{t_1} u_2 + 2 u_3) \partial_{t_1} + 2 (2 \partial_{t_1}^2 u_2 + 3\partial_{t_1} u_3 + 2 u_4 + 3 u_2^2). 
    \end{split}
\end{equation}
Thus, via the Lax equation (\ref{lax_equation}), we obtain the following ncKP hierarchy
\begin{subequations}
\begin{align}  \label{m=2}
  &  \partial_{t_2} L = \left[B_2, L\right] \Leftrightarrow
    \left\{
    \begin{array}{l}
        \partial_{t_2} u_2 = \partial_{t_1}^2 u_2 + 2\partial_{t_1} u_3, \\
        \partial_{t_2} u_3 = \partial_{t_1}^2 u_3 + 2\partial_{t_1} u_4 + 2 u_2 \partial_{t_1} u_2 + 2 \left[u_2, u_3\right], \\
        \partial_{t_2} u_4 = \partial_{t_1}^2 u_4 + 2\partial_{t_1} u_5 + 4 u_3 \partial_{t_1} u_2 - 2 u_2 \partial_{t_1}^2 u_2 + 2 \left[u_2, u_4\right], \\
        \quad \cdots,
    \end{array}
    \right.\\
  \label{m=3}
  &  \partial_{t_3} L = \left[B_3, L\right] \Leftrightarrow
    \left\{
    \begin{array}{l}
        \partial_{t_3} u_2 = \partial_{t_1}^3 u_2 + 3\partial_{t_1}^2 u_3 + 3\partial_{t_1} u_4 + 3\{u_2, \partial_{t_1} u_2\}, \\
        \partial_{t_3} u_3 = \partial_{t_1}^3 u_3 + 3\partial_{t_1}^2 u_4 + 3\partial_{t_1} u_5 + 3\partial_{t_1} u_2 u_3 + 6 u_2 \partial_{t_1} u_3 \\
        \quad\quad + 3 u_3 \partial_{t_1} u_2 + 3 [u_2, u_4], \\
        \quad \cdots,
    \end{array}
    \right.\\
      \label{m=4}
    &\partial_{t_4} L = \left[B_4, L\right] \Leftrightarrow
    \left\{
    \begin{array}{l}
        \partial_{t_4} u_2 = \partial_{t_1}^4 u_2 + 4\partial_{t_1}^3 u_3 + 6\partial_{t_1}^2 u_4 + 4\partial_{t_1} u_5 + 2 u_2 \partial_{t_1}^2 u_2 + 6\partial_{t_1} u_2 \partial_{t_1} u_2 \\
        \quad\quad + 4\partial_{t_1}^2 u_2 u_2 + 6\partial_{t_1} \{u_2, u_3\}, \\
        \quad \cdots.
    \end{array}
    \right.
\end{align}
\end{subequations}
In the above equations, we use notations $[A, B] = AB - BA$ and $\{A, B\} = AB + BA$ to denote commutator and anti-commutator, respectively.

In the subsequent, we show how to obtain some simple NC equations from the hierarchy. 
For example, we can get the ncKP equation \cite{gilson07,hamanaka03,hamanaka10,paniak01,sakakibara04}
\begin{equation} \label{ncKP}
    u_{t_3} = \frac{1}{4} u_{t_1t_1t_1} + \frac{3}{4} \partial_{t_1}^{-1} u_{t_2t_2} - \frac{3}{4} \left[u, \partial_{t_1}^{-1} u_{t_2}\right] + \frac{3}{4} \{u, u_{t_1}\},\quad \partial_{t_1}^{-1} f(t_1) = \int^{t_1} f(t_1^\prime) dt_1^\prime
\end{equation}
if we get the explicit expressions for  $u_4$ and $u_3$ from (\ref{m=2})
\begin{equation} \label{u_3}
    \partial_{t_1} u_4 = \frac{1}{2} \partial_{t_2} u_3 - \frac{1}{2} \partial_{t_1}^2 u_3 - u_2 \partial_{t_1} u_2 - \left[u_2, u_3\right], \quad \partial_{t_1} u_3 = \frac{1}{2} \partial_{t_2} u_2 - \frac{1}{2} \partial_{t_1}^2 u_2,
\end{equation}
and substitute them into the first equation in (\ref{m=3}) with variable transformations $u = 2 u_2$.


Besides, from the first two equations in (\ref{m=3}), we can obtain the expressions for $u_5$ and $u_4$
\begin{equation}
    \begin{split}
        &\partial_{t_1} u_5 = \frac{1}{3} \partial_{t_3} u_3 - \frac{1}{3} \partial_{t_1}^3 u_3 - \partial_{t_1}^2 u_4 - \partial_{t_1} u_2 u_3 - 2 u_2 \partial_{t_1} u_3 - u_3 \partial_{t_1} u_2 - [u_2, u_4], \\ 
        &\partial_{t_1} u_4 = \frac{1}{3} \partial_{t_3} u_2 - \frac{1}{3} \partial_{t_1}^3 u_2 - \partial_{t_1}^2 u_3 - \{u_2, \partial_{t_1} u_2\}.
    \end{split}
\end{equation} 
Substituting the above equations together with the expression for $u_3$ from equation (\ref{u_3}) into equation (\ref{m=4}), and setting $u = 2u_2$, we obtain
\begin{equation} \label{ncJM}
    u_{t_4} = \frac{1}{3} u_{t_1t_1t_2} + \frac{2}{3} \partial_{t_1}^{-1} u_{t_2t_3} + \frac{1}{6} [u, u_{t_1t_1}] + \frac{1}{2} \{u, u_{t_2}\} + \frac{1}{2} \{u_{t_1}, \partial_{t_1}^{-1} u_{t_2}\} - \frac{2}{3} [u, \partial_{t_1}^{-1} u_{t_3}], 
\end{equation}
which is the noncommutative version of the $(3 + 1)$-dimensional Date-Jimbo-Kashiwara-Miwa (ncDJKM) equation \cite{jimbo83}.  

\subsection{Reduction of ncKP}
In this section, we show how to construct NC equations in the sub-hierarchy of ncKP. 
The first and the most important example is to obtain the ncKdV equation by using the 2-reduction method, i.e. to set all derivatives of variables $u_k$ with respect to $t_{2n}$ $(n=1,2,\cdots)$ are identically to zero. From equation (\ref{m=2}), if we take $\partial_{t_2}u_2=\partial_{t_2}u_3=0$, then
we can express $u_3$ and $u_4$ in terms of $u_2$ by (c.f. equation \eqref{u_3})
\begin{equation}
    \partial_{t_1} u_3 = - \frac{1}{2} \partial_{t_1}^2 u_2, \quad \partial_{t_1} u_4 = \frac{1}{4} \partial_{t_1}^3 u_2 - \frac{1}{2} \{u_2, \partial_{t_1} u_2\}.
\end{equation}
Substituting the above equations into the first equation of (\ref{m=3}), we obtain
\begin{equation} \label{ncKdV}
    u_{t_3} = \frac{1}{4} u_{t_1t_1t_1} + \frac{3}{4} \{u, u_{t_1}\},
\end{equation} 
which is the noncommutative KdV equation with $u = 2u_2$ \cite{hamanaka05} .
Similarly, we can derive the fifth-order noncommutative KdV equation
\begin{equation}
    u_{t_5} = \frac{1}{16} u_{t_1t_1t_1t_1t_1} + \frac{5}{16} \{u, u_{t_1t_1t_1}\} + \frac{5}{8} (u_{t_1} u_{t_1} + u u u)_{t_1},
\end{equation}
where we have set $u = 2u_2$.

If we consider the $3$-reduction of ncKP, that is, to set $\partial_{t_3}u_2=0$, then in the first equation in (\ref{m=3}), we can get 
\begin{equation}
    \partial_{t_1} u_4 = - \frac{1}{3} \partial_{t_1}^3 u_2 - \partial_{t_1}^2 u_3 - \{u_2, \partial_{t_1} u_2\}.
\end{equation}
Substituting the above equation, along with the expression for $u_3$ from equation (\ref{u_3}), into the second equation in (\ref{m=2}), we obtain
\begin{equation} \label{ncBoussinesq}
    \partial_{t_1}^{-1} u_{t_2t_2} = - \frac{1}{3} u_{t_1t_1t_1} - \{u, u_{t_1}\} + \left[u, \partial_{t_1}^{-1} u_{t_2}\right],
\end{equation}
which is the noncommutative Boussinesq equation with $u = 2u_2$ \cite{hamanaka05}.

\section{Quasi-determinant Solution of NC Integrable Equations}

This section is devoted to the direct method for NC integrable equations by using quasi-determinants.  Hirota's direct method (or so-called bilinear method) tells us that if we take determinant solutions into bilinear equations, then the bilinear equations are exactly the Jacobi identities or Pl\"ucker relations, thereby demonstrating the validity of the solutions for integrable systems. In the NC case, we show that if we take the quasi-determinant solutions into the nonlinear equations, then an identity naturally appears. \emph{Although the equation satisfied by the quasi-determinant is no longer bilinear, the quasi-determinant plays the same role of $\tau$-function in non-commutative integrable system.} This method is important in two folds: (1) we can verify solutions for NC equations; (2) we can construct NC equations from quasi-determinant solutions. Verifications of solutions are common in literature, and in this part we show how to construct NC equations in Section 2 from quasi-determinant solutions.
\subsection{Preliminary: Equivalence between quasi-Wronskian and quasi-Grammian determinants}
There are two types of quasi-determinants used in literature to construct solutions for NC integrable equations. The first takes the form
\begin{equation}
    Q_n^{i, j} =
\left|
\begin{array}{cc}
    \widehat{\Theta} & e_{n-j} \\
    \Theta^{(n+i)} & \fbox{0}
\end{array}
\right|,
\end{equation}
where $e_k$ is the $n$-dimensional unit column vector with 1 in the $k$th row and 0 elsewhere. Moreover, $\Theta=(\theta_1,\cdots,\theta_n)$ is a row vector with seed functions $\theta_1,\cdots,\theta_n$, and $\widehat{\Theta} = \left( \theta_{j}^{(i-1)} \right)_{i,j=1,\dots,n}$ is an $n \times n$ Wronskian matrix with seed functions $\theta_1, \ldots, \theta_n$ satisfying $\theta_j^{(i)}=\partial_{t_1}^i \theta_j$.
The quasi-determinant $Q_n^{0, 0}$ is usually referred to as the quasi-Wronskian determinant, initially proposed in the construction of solutions for equations satisfying Marchenko lemma \cite{etingof97}, and later also used for Darboux transformation \cite{gelfand05}.

Despite the quasi-Wronskian determinant, we have another quasi-determinant called quasi-Grammian determinant, appeared in the process of binary Darboux transformation \cite{gilson07}. It admits the form
\begin{equation} \label{quasigrammian}
    R_n^{i, j} = (-1)^j
\left|
\begin{array}{cc}
    \Omega(\Theta, \Xi) & \Xi^{\dagger (j)} \\
    \Theta^{(i)} & \fbox{0}
\end{array}
\right|,
\end{equation}
where $\Omega(\Theta, \Xi) = \partial_{t_1}^{-1} (\Xi^{\dagger} \Theta)$ and $\Xi^{\dagger}$ denotes the complex conjugate of $\Xi=(\xi_1,\cdots,\xi_n)$. 

\begin{proposition}\label{prop1}
If seed functions $\theta_1,\cdots,\theta_n$ and $\xi_1,\cdots,\xi_n$ satisfy dispersion relations
\begin{align}\label{theta}
 \frac{\partial \theta_j}{\partial t_m} = \frac{\partial^m \theta_j}{\partial {t_1}^m},\quad \frac{\partial \xi_j}{\partial{t_m}}=(-1)^{m-1} \frac{\partial^m \xi_j}{\partial {t_1}^m},
 \quad j=1,\cdots,n,\,m=1,2,\cdots,
\end{align}
then the quasi-Wronskian $Q_n^{i, j}$ and the quasi-Grammian $R_n^{i, j}$ satisfy the same derivative rules.
\end{proposition}
\begin{proof}
For any $m$, using linear relations (\ref{theta}) and the derivative formula (\ref{dquasiwronskian}), we obtain
\begin{align*}
    \frac{\partial}{\partial t_m} Q_n^{i, j} 
    &= \left| \begin{array}{cc}
    \widehat{\Theta} & e_{n-j} \\
    \Theta^{(n+i+m)} & \boxed{0}
    \end{array} \right| 
    + \sum_{k=1}^{n} 
    \left| \begin{array}{cc}
    \widehat{\Theta} & e_k \\
    \Theta^{(n+i)} & \boxed{0}
    \end{array} \right|
    \left| \begin{array}{cc}
    \widehat{\Theta} & e_{n-j} \\
    \Theta^{(k-1+m)} & \boxed{0}
    \end{array} \right| \\
    &= Q_n^{i + m, j} + \sum_{k=0}^{n-1} Q_n^{i,k} Q_n^{m-1-k,j}.
\end{align*}
For $k<m-1$ and $m-k+j\neq0$, it is easily checked that $Q_n^{m-1-k,j}=0$, and for $m-k+j=0$, we have 
$Q_n^{m-1-k,j}=-1$, then 
\begin{equation}\label{deri_quasi}
    \frac{\partial}{\partial t_m} Q_n^{i, j} = Q_n^{i + m, j} - Q_n^{i, j + m} + \sum_{k=0}^{m-1} Q_n^{i, k} Q_n^{m - k - 1, j}.
\end{equation}
On the other hand, by using (\ref{dquasigrammian}), the derivative of $R_n^{i, j}$ with respect to $t_m$ could be computed as
\begin{equation}
    \begin{split}
    \frac{\partial}{\partial t_m} R_n^{i, j} &= (-1)^j \begin{vmatrix}
        \Omega & \Xi^{\dagger(j)} \\
        \Theta^{(i+m)} & \boxed{0} 
        \end{vmatrix} + (-1)^{m+j-1} \begin{vmatrix}
        \Omega & \Xi^{\dagger(j+m)} \\
        \Theta^{(i)} & \boxed{0} 
        \end{vmatrix} \\
        &\quad\quad + \sum_{k=0}^{m-1} \begin{vmatrix}
        (-1)^{j+k} \Omega & \Xi^{\dagger(k)} \\
        \Theta^{(i)} & \boxed{0} 
        \end{vmatrix} \begin{vmatrix}
        \Omega & \Xi^{\dagger(j)} \\
        \Theta^{(m-1-k)} & \boxed{0} 
        \end{vmatrix} \\
        &\quad = R_n^{i+m, j} - R_n^{i, j+m} + \sum_{k=0}^{m-1} R_n^{i, k}R_n^{m-k-1, j}.
    \end{split}
\end{equation}
\end{proof}
It is important to note that although the derivative rules for the quasi-Wronskian and quasi-Grammian are identical, the corresponding seed functions $\Theta$ and $\Xi$ may be different.
However, according to the Proposition \ref{prop1}, the construction of NC integrable systems from the quasi-Wronskian solution is enough. The equations satisfied by quasi-Wronskian are equally valid for the quasi-Grammian solutions. In the subsequent, we construct ncKP equations from $Q_n^{0, 0}$.

\subsection{Construction of ncKP from quasi-determinants}
In the commutative case, the construction of the KP equation via the $\tau$-function ultimately yields a determinant identity (namely, the Plücker relation or Jacobi identity). In contrast, within the noncommutative framework, employing the quasi-Wronskian $Q_n^{0, 0}$ to derive the ncKP equation leads to a complete cancellation of all $Q$-terms. This phenomenon substantiates the observation made by Gilson and Nimmo, wherein the identities required for computing derivatives of quasi-determinants are inherently encoded within the structure of the quasi-determinants themselves. Consequently, our approach to formulating the ncKP equation begins with identifying a linear combination of derivatives of $Q_n^{0, 0}$ that vanishes identically. For notational convenience, we set $v = Q_n^{0, 0}$. 

To facilitate this procedure, we begin by computing the derivatives of $v$, see Appendix B. 
We know the derivatives $v_{t_1t_3}$, $v_{t_2t_2}$, and $v_{t_1t_1t_1t_1}$ contain monomials $Q_n^{i, j}$ up to $i+j=4$, and we list the coefficient below. 
\begin{equation} \nonumber
    \renewcommand{\arraystretch}{1.25}
    \begin{array}{|c|c|c|c|}
        \hline
         & v_{t_1t_3} & v_{t_2t_2} & v_{t_1t_1t_1t_1} \\ \hline
        Q_n^{4, 0}(Q_n^{0, 4}) & 1 & 1 & 1 \\ \hline
        Q_n^{2, 2} & 0 & -2 & 6 \\ \hline
        Q_n^{3, 1}(Q_n^{1, 3}) & -1 & 0 & -4 \\ \hline
    \end{array}
\end{equation}
In order to eliminate these terms, it is necessary to construct an appropriate linear combination $x_1v_{t_1t_3}+x_2v_{t_2t_2}+x_3v_{t_1t_1t_1t_1}$, which amounts to solving the following system of linear equations
\begin{equation} \nonumber
\left\{
\begin{array}{rrrrrl}
x_1 & +\; x_2 & +\;  x_3 & =\; 0, \\
x_1 & \;   & +\; 4x_3 & =\; 0, \\
& \; 2x_2 & -\;  6x_3 & =\; 0.
\end{array}
\right.
\end{equation} 
Solving the above system yields $(x_1,x_2,x_3)=(-4,3,1)$, and we have
\begin{align} \hspace{-2cm} \label{part_1}
\begin{aligned}
   & -4v_{t_1t_3} + 3v_{t_2t_2} + v_{t_1t_1t_1t_1} \\
    &\quad= 12[Q_n^{2, 0} - Q_n^{1, 1}][Q_n^{1, 0} - Q_n^{0, 1}] + 12[Q_n^{1, 0} - Q_n^{0, 1}][Q_n^{0, 2} - Q_n^{1, 1}] + \cdots.
    \end{aligned}
\end{align}
Examining the right-hand side of the above equation, we find that in order to further eliminate nonlinear terms such as $Q_n^{1, 0}Q_n^{2, 0}$ and $Q_n^{2, 0}Q_n^{1, 0}$, it is necessary to consider an appropriate linear combination of $v_{t_1} v_{t_1t_1}$, $v_{t_1} v_{t_2}$, $v_{t_1t_1} v_{t_1}$, and $v_{t_2} v_{t_1}$, for simplicity, we denote $-4v_{t_1t_3} + 3v_{t_2t_2} + v_{t_1t_1t_1t_1}$ by $\star$
\begin{align*} 
    \begin{array}{|c|c|c|c|}
        \hline
        & v_{t_1} v_{t_1t_1} & v_{t_1} v_{t_2} & \star \\ \hline
        Q_n^{1, 0}Q_n^{2, 0}(Q_n^{0, 1}Q_n^{2, 0}) & -1(1) & -1(1) & 0 \\ \hline
        Q_n^{1, 0}Q_n^{1, 1}(Q_n^{0, 1}Q_n^{1, 1}) & -2(2) & 0 & -12(12) \\ \hline
        Q_n^{1, 0}Q_n^{0, 2}(Q_n^{0, 1}Q_n^{0, 2}) & 1(-1) & -1(1) & 12(-12) \\ \hline
    \end{array} 
    \end{align*}
    \begin{align*}
    \begin{array}{|c|c|c|c|}
        \hline
        & v_{t_1t_1} v_{t_1} & v_{t_2} v_{t_1} & \star \\ \hline
        Q_n^{2, 0}Q_n^{1, 0}(Q_n^{2, 0}Q_n^{0, 1}) & 1(-1) & 1(-1) & 12(-12) \\ \hline
        Q_n^{1, 1}Q_n^{1, 0}(Q_n^{1, 1}Q_n^{0, 1}) & -2(2) & 0 & -12(12) \\ \hline
        Q_n^{0, 2}Q_n^{1, 0}(Q_n^{0, 2}Q_n^{0, 1}) & 1(-1) & -1(1) & 0 \\ \hline
    \end{array}
\end{align*}
It is readily seen that adding the term $-6v_{t_1}(v_{t_1t_1} - v_{t_2}) - 6(v_{t_1t_1} + v_{t_2})v_{t_1}$ to the left-hand side of equation (\ref{part_1}) eliminates the nonlinear terms listed in the table above. Fortunately, the right-hand side vanishes in this case, and thus we obtain
\begin{equation} \label{nckp}
    -4v_{t_1t_3} + 3v_{t_2t_2} + v_{t_1t_1t_1t_1} - 6\{v_{t_1}, v_{t_1t_1}\} + 6[v_{t_1}, v_{t_2}] = 0,
\end{equation}
which is just the ncKP equation (\ref{ncKP}) with $-2v_{t_1} = u$. 

Furthermore, the same procedure can be systematically applied to derive higher-order noncommutative KP equations. We know the derivatives $v_{t_1t_1t_1t_2}$, $v_{t_2t_3}$, and $v_{t_1t_4}$ contain monomials $Q_n^{i, j}$ up to $i+j=5$, and we list the coefficient below.  
\begin{equation} \nonumber
    \renewcommand{\arraystretch}{1.25}
    \begin{array}{|c|c|c|c|}
        \hline
         & v_{t_1t_1t_1t_2} & v_{t_2t_3} & v_{t_1t_4} \\ \hline
        Q_n^{5, 0}(Q_n^{0, 5}) & 1 & 1 & 1 \\ \hline
        Q_n^{4, 1}(Q_n^{1, 4}) & -3 & 0 & -1 \\ \hline
        Q_n^{3, 2}(Q_n^{2, 3}) & 2 & -1 & 0 \\ \hline
    \end{array}
\end{equation}
In order to eliminate these terms, it is necessary to construct an appropriate linear combination $x_1v_{t_1t_1t_1t_2}+x_2v_{t_2t_3}+x_3v_{t_1t_4}$. Thus, we have the following linear equations
\begin{equation} \nonumber
    \left\{
    \begin{array}{rrrrrl}
    x_1 & +\; x_2 & +\;  x_3 & =\; 0, \\
    3x_1 & \;   & +\; x_3 & =\; 0, \\
    2x_1 & -\; x_2 & \;   & =\; 0.
    \end{array}
    \right.
\end{equation}
Solving the above system, we obtain
\begin{align}  \label{part_2}
    \begin{aligned}
       &v_{t_1t_1t_1t_2} + 2 v_{t_2t_3} - 3v_{t_1t_4} = 6[Q_n^{3, 0} - Q_n^{1, 2}][Q_n^{1, 0} - Q_n^{0, 1}] - 6[Q_n^{1, 0} - Q_n^{0, 1}][Q_n^{2, 1} - Q_n^{0, 3}] \\
       & \qquad\qquad\qquad\qquad\quad - 6\{Q_n^{1, 1}, Q_n^{2, 0} - Q_n^{0, 2}\} + 6Q^{2, 0}Q_n^{2, 0} - 6Q^{0, 2}Q_n^{0, 2} + \cdots. 
    \end{aligned}
\end{align}
To eliminate the nonlinear terms on the right-hand side of the above equation, such as 
$Q_n^{3, 0}Q_n^{1, 0}$, $Q_n^{0, 1}Q_n^{0, 3}$, and $Q_n^{2, 0}Q_n^{0, 2}$, we need to consider the linear combination of $v_{t_1t_1t_1} v_{t_1}$, $v_{t_1t_2} v_{t_1}$, $v_{t_3} v_{t_1}$, $v_{t_1} v_{t_1t_1t_1}$, $v_{t_1} v_{t_1t_2}$, $v_{t_1} v_{t_3}$, $v_{t_1t_1} v_{t_2}$, and $v_{t_2} v_{t_1t_1}$, for simplicity, we denote $v_{t_1t_1t_1t_2} + 2 v_{t_2t_3} - 3v_{t_1t_4}$ by $\star\star$, and we have
\begin{align*}
    \begin{array}{|c|c|c|c|c|}
        \hline
         & v_{t_1t_1t_1} v_{t_1} & v_{t_1t_2} v_{t_1} & v_{t_3} v_{t_1} & \star\star \\ \hline
        Q_n^{3, 0}Q_n^{1, 0}(Q_n^{3, 0}Q_n^{0, 1}) & 1(-1) & 1(-1) & 1(-1) & 6(-6) \\ \hline
        Q_n^{2, 1}Q_n^{1, 0}(Q_n^{2, 1}Q_n^{0, 1}) & -3(3) & -1(1) & 0 & 0 \\ \hline
        Q_n^{1, 2}Q_n^{1, 0}(Q_n^{1, 2}Q_n^{0, 1}) & 3(-3) & -1(1) & 0 & -6(6) \\ \hline
        Q_n^{0, 3}Q_n^{1, 0}(Q_n^{0, 3}Q_n^{0, 1}) & -1(1) & 1(-1) & -1(1) & 0 \\ \hline
    \end{array}
\end{align*}
\begin{align*}
    \begin{array}{|c|c|c|c|c|}
        \hline
         & v_{t_1} v_{t_1t_1t_1} & v_{t_1} v_{t_1t_2} & v_{t_1} v_{t_3} & \star\star \\ \hline
        Q_n^{1, 0}Q_n^{3, 0}(Q_n^{0, 1}Q_n^{3, 0}) & 1(-1) & 1(-1) & 1(-1) & 0 \\ \hline
        Q_n^{1, 0}Q_n^{2, 1}(Q_n^{0, 1}Q_n^{2, 1}) & -3(3) & -1(1) & 0 & -6(6) \\ \hline
        Q_n^{1, 0}Q_n^{1, 2}(Q_n^{0, 1}Q_n^{1, 2}) & 3(-3) & -1(1) & 0 & 0 \\ \hline
        Q_n^{1, 0}Q_n^{0, 3}(Q_n^{0, 1}Q_n^{0, 3}) & -1(1) & 1(-1) & -1(1) & 6(-6) \\ \hline
    \end{array} 
\end{align*}
\begin{align*}
    \begin{array}{|c|c|c|c|c|}
        \hline
         & v_{t_1t_1} v_{t_2} & v_{t_2} v_{t_1t_1} & \star\star \\ \hline
        Q_n^{2, 0}Q_n^{2, 0}(Q_n^{0, 2}Q_n^{0, 2}) & 1(-1) & 1(-1) & 6(-6) \\ \hline
        Q_n^{2, 0}Q_n^{0, 2}(Q_n^{0, 2}Q_n^{2, 0}) & -1(1) & 1(-1) & 0 \\ \hline
        Q_n^{1, 1}Q_n^{2, 0}(Q_n^{1, 1}Q_n^{0, 2}) & -2(2) & 0 & -6(6) \\ \hline
        Q_n^{2, 0}Q_n^{1, 1}(Q_n^{0, 2}Q_n^{1, 1}) & 0 & -2(2) & -6(6) \\ \hline
    \end{array}
\end{align*}
From the table, we observe that the right-hand side of equation (\ref{part_2}) can be eliminated by the linear combination $(v_{t_1t_1t_1} - 3v_{t_1t_2} - 4v_{t_3})v_{t_1} - v_{t_1}(v_{t_1t_1t_1} + 3v_{t_1t_2} - 4v_{t_3}) - 3\{v_{t_1t_1}, v_{t_2}\}$. As a result, we arrive at an equation with a vanishing right-hand side
\begin{align} 
        v_{t_1t_1t_1t_2} + 2 v_{t_2t_3} - 3v_{t_1t_4} - [v_{t_1}, v_{t_1t_1t_1}] - 3\{v_{t_1}, v_{t_1t_2}\} + 4[v_{t_1}, v_{t_3}] - 3\{v_{t_1t_1}, v_{t_2}\} = 0, 
\end{align}
which is just the ncDJKM equation (\ref{ncJM}) with $-2v_{t_1} = u$. A broader class of higher-order equations in the ncKP hierarchy can also be derived using the same method. To conclude, we can have the following theorem.
\begin{theorem}\label{thm_ncKP}
If seed functions $\theta_i$ satisfy dispersion relations 
\begin{align}\label{dispersion}
\frac{\partial \theta_i}{\partial t_m}=\frac{\partial^m \theta_i}{\partial t_1^m}, 
\end{align}
then the quasi-Wronskian determinant $Q_n^{0,0}=\left|
\begin{array}{cc}
    \widehat{\Theta} & e_{n} \\
    \Theta^{(n)} & \boxed{0}
\end{array}
\right|$ is a solution of ncKP hierarchy. 
\end{theorem}

\subsection{Matrix Solutions of ncKP}
In the Theorem \ref{thm_ncKP}, we show the solution of ncKP could be expressed in terms of quasi-Wronskian determinants. In this part, we show how to take proper seed functions to get the one- and two-soliton matrix solutions of the ncKP equation.

Similar to the commutative case, we can take the seed functions for soliton solutions admitting the expression
\begin{align*}
\theta_i=1+\alpha_i e^{k_it_1+k_i^2t_2+k_i^3t_3+\eta_i^0},
\end{align*}
where $\alpha_i,\,k_i$ and $\eta_i^0$ are elements in a skew field and $1$ is the unit in the skew field. 
Obviously, these seed functions satisfy dispersion relations \eqref{dispersion}.
In the case $n=1$, we obtain a noncommutative one-soliton solution. Expanding the quasi-determinant gives
\begin{align*}
        v = \left|\begin{array}{cc}
        \theta_1&1\\
        \theta_1'&\boxed{0}
        \end{array}
        \right|=-\theta_1'\theta_1^{-1}=-\alpha_1k_1e^{k_1t_1+k_1^2t_2+k_1^3t_3+\eta_1^0}(1+\alpha_1e^{k_1t_1+k_1^2t_2+k_1^3t_3+\eta_1^0})^{-1}.
\end{align*}        
This expression can hardly be simplified if we don't take values for coefficients $\alpha_1$ and $k_1$, and $\eta_1^0$. A special choice of the seed function is consider a matrix-valued representation for the 
division ring, such as matrix representation for quaternions. For example, 
we can take seed function $\theta_i$ as a $d\times d$ dimensional matrix-valued function which admits the form\footnote{We choose this seed function to coincide our results with the one given in  \cite[Eq.18]{gilson09} although their results are expressed in terms of quasi-Grammian determinant. Of course, we can also consider the seed function $\theta_i=I_d+r_iP_ie^{k_it_1+k_i^2t_2+k_i^3t_3+\eta_i^0}$.}
\begin{equation} \label{seed-kp}
    \theta_i = r_i P_i \mathrm{e}^{ k_i t_1 + k_i^2 t_2 + k_i^3 t_3 + \eta_i^0} + I_d \mathrm{e}^{q_i t_1 + q_i^2 t_2 + q_i^3 t_3 + \gamma_i^0}
\end{equation}
where  $r_i, \, k_i,\,q_i,\,\eta_0$ and $\gamma_0$ are scalars, $P_i$ is a $d \times d$ projection matrix, and $I_d$ denotes the $d \times d$ identity matrix.
Therefore, we can get a matrix-valued 1-soliton solution
\begin{equation}
    \begin{split}
        v &= \left|\begin{array}{cc}
        \theta_1&I_d\\
        \theta_1'&\boxed{0_d}
        \end{array}
        \right|=- \left(q_1 I_d + k_1r_1P_1 \mathrm{e}^{\eta_1-\gamma_1}\right) \left(I_d + r_1P_1\mathrm{e}^{\eta_1-\gamma_1}\right)^{-1} \\
        &= -\frac{r_1 \mathrm{e}^{\eta_1-\gamma_1}}{1 + r_1 \mathrm{e}^{\eta_1-\gamma_1}} \left(k_1-q_1\right)P_1 - q_1I_d.
    \end{split}
\end{equation}
where we used simplified notations $\gamma_1=q_1t_1+q_1^2t_2+q_1^3t_3+\gamma_1^0$ and $\eta_1=k_1t_1+k_1^2t_2+k_1^3t_3+\eta_1^0$, and $0_d$ is the $d\times d$ null matrix.
Besides, we used the property of projective matrix $P^2=P$, which induces the formula $\left(I - aP\right)^{-1} = I +a(1-a)^{-1} P$, by assuming $a \neq 1$. Then the matrix-valued 1-soliton solution of (\ref{ncKP}) is given by
\begin{equation} \label{1-soliton}
    u = -2v_{t_1} = \frac{1}{2} \left(k_1-q_1\right)^2 P_1 \mathrm{sech}^2 \left(\frac{\eta_1-\gamma_1+ \mathrm{log} r_1}{2}\right),
\end{equation}
as a matrix-valued generalization of the original sech function solution.

For a matrix-valued 2-soliton solution,  we could use a similar argument, from which we have
\begin{equation}
    \begin{split}
    v &= \left|\begin{array}{ccc}
        \theta_1&\theta_2&0_d\\
        \theta_1'&\theta_2'&I_d\\
       \theta_1''&\theta_2''& \boxed{0_d}
        \end{array}
        \right|=-\begin{pmatrix} \theta_1'' & \theta_2''\end{pmatrix}\begin{pmatrix} \theta_1&\theta_2 \\ \theta_1'&\theta_2'\end{pmatrix}^{-1} \begin{pmatrix} 0_d \\ I_d \end{pmatrix} \\
    & = - \begin{pmatrix} K_1 & K_2 \end{pmatrix} \begin{pmatrix} 0_d \\ I_d \end{pmatrix} = - K_2,
    \end{split}
\end{equation}
 where $K_1$ and $K_2$ satisfy
\begin{equation}
    \begin{split}
        &K_1\theta_1 + K_2\theta_1' = \theta_1'', \\
        &K_1\theta_2 + K_2\theta_2'= \theta_2''.
    \end{split}
\end{equation}
Solving for $K_2$ gives
\begin{align} \label{2-soliton}
\begin{aligned}
    K_2 &=(\theta_1''\theta_1^{-1}-\theta_2''\theta_2^{-1})(\theta_1'\theta_1^{-1}-\theta_2'\theta_2^{-1})^{-1}
    \\
   &=  \frac{\left[k_2\left(q_2^2-q_1^2\right)-\alpha k_2\left(k_1^2-q_1^2\right)+\beta a_2\left(k_1^2+q_2k_2\right)\right]I_d-\beta a_2\left(k_1^2-q_1^2\right)\left(I_d-\alpha P_1\right)P_2}{k_2\left(q_2-q_1\right)-\alpha k_2a_1+\beta q_1 a_2+\alpha \beta a_1a_2},
    \end{aligned}
\end{align}
where 
\begin{equation}
    \alpha = \frac{r_1 \mathrm{e}^{\eta_1-\gamma_1}}{1+r_1 \mathrm{e}^{\eta_1-\gamma_1}}, \qquad \beta =  \frac{r_2k_2 \mathrm{e}^{\eta_2-\gamma_2}}{q_2+r_2k_2 \mathrm{e}^{\eta_2-\gamma_2}}, \qquad a_i = k_i-q_i.
\end{equation}
Therefore, the matrix-valued 2-soliton solution of equation (\ref{ncKP}) is then given by $u = -2v_{t_1} = 2 \partial_{t_1} K_2$.

In fact, the matrix-valued soliton solutions of the ncDJKM equation share the same structures as those of the ncKP equation. By appropriately modifying $\eta_i$ and $\gamma_i$, equations (\ref{1-soliton}) and (\ref{2-soliton}) yield the matrix-valued 1- and 2-soliton solutions for the ncDJKM equation, respectively.

\section{Direct reduction method from ncKP hierarchy}
In this section, we demonstrate that by making appropriate constraints to seed functions, 
some sub-hierarchy of ncKP hierarchy would be obtained, directly from the constraints of solution for nonlinear equations. 
We mainly consider the 2-reduction of ncKP hierarchy (i.e. ncKdV hierarchy) and 3-reduction of ncKP hierarchy (i.e. ncBoussinesq hierarchy). 

\subsection{2-Reduction of ncKP}
If $v = Q_n^{0, 0}$ independent of the $t_2$-flow, then $v_{t_2}$ vanishes, and the ncKP equation (\ref{nckp}) can be seen to reduce to the ncKdV equation
\begin{equation}
    -4v_{t_1t_3} + v_{t_1t_1t_1t_1} - 6\{v_{t_1}, v_{t_1t_1}\} = 0.
\end{equation}
This is consistent with equation (\ref{ncKdV}) upon identifying $-2v_{t_1} = u$.

Therefore, to get the solution for ncKdV equation as well as ncKP equation, it is necessary to impose some constraints into the seed functions such that the derivative of $Q_n^{0, 0}$ with respect to $t_2$ vanishes. Such a result is stated as the following theorem. 
\begin{theorem} \label{thm1}
    If the seed functions $\theta_1,\cdots,\theta_n$ satisfy dispersion relations \eqref{dispersion}
   and the constraint 
   \begin{align}\label{dis_kdv}
   \partial_{t_1}^2 \theta_i=\theta_i k_i^2,
   \end{align} then
    we have $\partial_{t_2}Q_n^{0,0}=0$.
Here $k_i$ is a right multiplier in a skew field. 
\end{theorem}

\begin{proof}
  According to the derivative formula for quasi-determinants, we have
    \begin{equation} \label{y-derivative}
        \begin{split}
       \partial_{t_2} Q_n^{0,0} = - \Theta^{(n+2)} \widehat{\Theta}^{-1} e_{n} + \Theta^{(n)} \widehat{\Theta}^{-1} \widehat{\Theta}_{t_2} \widehat{\Theta}^{-1} e_{n},
        \end{split}
    \end{equation}
    Since $\partial_{t_1}^2 \theta_i = \theta_i k_i^2$, we get
    \[
        \Theta^{(n+2)} = \Theta^{(n)} \operatorname{diag}(k_1^2, k_2^2, \dots, k_n^2),\quad
    \widehat{\Theta}_{t_2} = \widehat{\Theta} \operatorname{diag}(k_1^2, k_2^2, \dots, k_n^2).    
    \]
   By substituting this into (\ref{y-derivative}), we obtain $\partial_{t_2} Q_n^{0,0} = 0$.
\end{proof}
It should be remarked that the constraint \eqref{dis_kdv} is not unique. We can also assume that seed functions $\theta_i$ satisfy a more general condition
\begin{align*}
\partial_{t_1}^2\theta_i=\sum_{j=1}^n \theta_j\lambda_{ji},\quad i=1,\cdots,n
\end{align*}
with $\lambda_{ji}$ being elements in a skew field. With this constraint, we conclude that $\partial_{t_2} Q_n^{0,0}=0$ by replacing the diagonal matrix by the matrix $\Lambda=(\lambda_{ij})_{i,j=1}^n$. This kind of constraint on seed functions has been considered in \cite{ma05} for finding different solutions in the commutative KdV equation.
Besides, we have the following statement for the solution of ncKdV hierarchy.
\begin{coro}
The ncKdV hierarchy (the ncKP 2-reduction hierarchy) admits a quasi-Wronskian solution $Q_n^{0,0}=\left|
\begin{array}{cc}
    \widehat{\Theta} & e_{n} \\
    \Theta^{(n)} & \boxed{0}
\end{array}
\right|$ with seed functions satisfying 
\begin{align*}
\partial_{t_{2n+1}}\theta_i=\partial_{t_1}^{2n+1}\theta_i,\quad \partial_{t_1}^2\theta_i=\sum_{j=1}^n \theta_i\lambda_{ji}.
\end{align*}
\end{coro}

\subsubsection{One-Soliton Solution of ncKdV equation} 
The one-soliton solution of ncKdV equation could be found in many references, such as \cite{paniak01,etingof97} and many others. We show that the solution given in earlier references could be embedded into the Theorem \ref{thm1}. For example, in the work \cite[Appendix A]{etingof97}, the authors showed that a 1-soliton solution of ncKdV could be generated by the seed functions
\begin{align*}
\theta_i=e^{k_i t_1+k_i^3t_3}+a_se^{-k_it_1-k_i^3t_3},
\end{align*}
where $a_i,\,k_i$ are some elements in a skew field. 
Obviously, this seed function satisfies the dispersion relation \eqref{dis_kdv}, then we have the quasi-determinant
\begin{align*}
Q_1^{0,0}=-\theta_1'(\theta_1)^{-1}=-(e^{k_1t_1+k_1^3t_3}-a_1e^{-k_1t_1-k_1^3t_3})k_1(e^{k_1t_1+k_1^3t_3}+a_1e^{-k_1t_1-k_1^3t_3})^{-1}.
\end{align*}
The solution of ncKdV equation is then given by $u=-2\partial_{t_1} Q_1^{0,0}$, which, in commutative case, reduces to the well-known soliton solution given by exponential function 
\begin{align*}
u=\frac{2k_1^2}{\cosh^2(k_1t_1+k_1^3t_3-\frac{1}{2}\ln k_1)}.
\end{align*}

\subsubsection{A matrix-valued 2-soliton solution for ncKdV equation}
Now we turn to a specific choice of seed function as a matrix-valued function to compute a 2-soliton solution for the ncKdV equation. 
Following the matrix-valued seed function chosen for ncKP equation, let's consider
\[
    \theta_i = I_d \mathrm{e}^{\zeta_i} + r_i P_i \mathrm{e}^{-\zeta_i},
\]
where $\zeta_i = k_i t_1 + k_i^3 t_3 + \zeta_i^0$, $k_i$ is a scalar, $r_i$ and $P_i$ are the same as those in (\ref{seed-kp}). Expanding the quasi-determinant expression, we have
\begin{equation}
    \begin{split}
    v &= -\begin{pmatrix} \theta_1k_1^2& \theta_2k_2^2 \end{pmatrix}  \begin{pmatrix} \theta_1&\theta_2 \\\theta_1'&\theta_2'\end{pmatrix}^{-1} \begin{pmatrix} 0_d \\ I_d \end{pmatrix} = - \begin{pmatrix} S_1 & S_2 \end{pmatrix} \begin{pmatrix} 0_d \\ I_d \end{pmatrix} = - S_2,
    \end{split}
\end{equation}
where $S_1$ and $S_2$ satisfy
\begin{equation}
    \begin{split}
        &S_1\theta_1+S_2\theta_1'=\theta_1k_1^2, \\
        &S_1\theta_2+S_2\theta_2'=\theta_2k_2^2.
    \end{split}
\end{equation}
Solving for $S_2$ gives
\begin{equation}
    S_2 =(k_1^2-k_2^2)(\theta_1'\theta_1^{-1}-\theta_2'\theta_2^{-1})^{-1}= \frac{(k_2^2-k_1^2)(I_d+2\beta_2P_2)}{4k_1\beta_1\beta_2-2k_1(\beta_2-\beta_1) + k_2-k_1},
\end{equation}
where 
\begin{equation}
    \beta_1 = \frac{r_1 \mathrm{e}^{-2\zeta_1}}{1+r_1 \mathrm{e}^{-2\zeta_1}}, \quad \beta_2 = \frac{r_2 \mathrm{e}^{-2\zeta_2}}{1-r_2 \mathrm{e}^{-2\zeta_2}}. 
\end{equation}
Therefore, the two-soliton matrix solution of equation (\ref{ncKdV}) is given by $u = -2v_{t_1} = 2 \partial_{t_1} S_2$.

\subsection{3-reduction of ncKP}
In the case of the ncBoussinesq equation, following the approach of the previous subsection, we assume that $v = Q_n^{0, 0}$ is independent of the variable $t_3$. Under this assumption, we have $\partial_{t_3} v=0$, and the ncKP equation (\ref{ncKP}) reduces to the ncBoussinesq equation
\begin{equation}
    3v_{t_2t_2} + v_{t_1t_1t_1t_1} - 6\{v_{t_1}, v_{t_1t_1}\} + 6[v_{t_1}, v_{t_2}] = 0,
\end{equation}
which is consistent with equation (\ref{ncBoussinesq}) under the identification $-2v_{t_1} = u$. 
Similar to the 2-reduction case, we have the following theorem. 
\begin{theorem}
    If the seed functions $\theta_1,\cdots,\theta_n$ satisfy dispersion relations \eqref{dispersion} and the constraint 
    \begin{equation} \label{seed-Boussinesq}
    \partial_{t_1}^3 \theta_j=\theta_j\alpha_j,\quad j=1,2,\cdots,n,
    \end{equation}
   then we have $\partial_{t_3} Q_n^{0,0}=0$.
Here, $\alpha_j$ is a right multiplier of a skew field.
\end{theorem}
As mentioned in \cite[Section 4]{kodama24}, the above ncBoussinesq equation doesn't admits any soliton solution, and thus we need to consider a modified constraint to the seed function. 
A suitable modification could be given by 
\begin{align}\label{cons_bsq}
\partial_{t_1}^3\theta_j=c_0\partial_{t_1} \theta_j+\sum_{i=1}^n \theta_j\lambda_{ji},\quad j=1,\cdots,n,
\end{align}
where $\lambda_{ji}$ are elements in a skew field and $c_0$ is a centre of the skew field.
In this case, we can get $\partial_{t_3}Q_n^{0,0}=c_0\partial_{t_1} Q_n^{0,0}$, and then $v=Q_n^{0,0}$ satisfies the equation
\begin{align*}
3v_{t_2t_2}+v_{t_1t_1t_1t_1}-c_0v_{t_1t_1}-6\{v_{t_1},v_{t_1t_1}\}+6[v_{t_1},v_{t_2}]=0,
\end{align*}
which is the standard form of ncBoussinesq equation.

In the commutative case, the constraint \eqref{cons_bsq} with $c_0=-3/4$ has been considered in \cite{li07}. Different kinds of solutions, such as solitons, negatons, positons and complexitons has been discussed. Those results could be easily generalized into the matrix-valued, and we omit the details here.

\section{Conclusion and Discussion}

In this work, we have developed a systematic method for constructing NC integrable equations from quasi-determinant solutions, thereby generalizing the classical direct method to the noncommutative realm. Beginning with the quasi-Wronskian form, we demonstrated that the ncKP and ncDJKM equations arise as derivative identities of quasi-determinants, rather than requiring ad hoc verification. The derived equations exhibit complete cancellation of nonlinear terms, reflecting the intrinsic algebraic coherence of the quasi-determinant framework.

We further explored integrable reductions of the ncKP hierarchy by imposing constraints on the flow variables. In particular, we showed that the ncKdV and ncBoussinesq hierarchies emerge naturally when the seed functions satisfy suitable dispersion relations. Beyond the specific equations considered here, our framework offers a pathway for deriving broader classes of NC integrable systems and their solutions. It would be interesting to extend this method to noncommutative analogues of other hierarchies (such as non-commutative CKP hierarchy) and investigate connections with NC gauge theories. The quasi-determinant, in this context, stands not only as a computational device but as a unifying algebraic principle within the theory of NC integrability.

\section*{Acknowledgement}
This work is partially funded by grants (Grant No. NSFC 12371251, 12175155). G.F. Yu was also partially supported by the National Key Research and Development Program of China (Grant No. 2024YFA1014101), Shanghai Frontier Research Institute for Modern Analysis and the Fundamental Research Funds for the Central Universities.

\section*{Data Availability Statement}
There is no data needed in this research.

\section*{Interest Statement}
There is no conflict of interest.

\appendix
\section{Introduction to Quasi-determinants}
\noindent In the appendix, we briefly recall some of the properties of quasi-determinants. The detailed discussion is seen in \cite{gelfand91,gelfand05}.

Let $A = (a_{ij})$ be an $n \times n$ matrix with entries over a usually noncommutative unitary ring $\mathcal{R}$. We denote the $(i, j)^{\text{th}}$ quasi-determinant by $\left| A \right|_{ij}$, where
\begin{equation} \nonumber
    \left| A \right|_{ij} = \begin{vmatrix}
        A^{ij} & c_j^i \\
        r_i^j & \boxed{a_{ij}}
        \end{vmatrix} = a_{ij} - r_i^j (A^{ij})^{-1} c_j^i.
\end{equation}
Here, $A^{ij}$ is the $(n-1) \times (n-1)$ minor matrix obtained from $A$ by deleting the $i^{\text{th}}$ row and $j^{\text{th}}$ column (note that this matrix must be invertible), $r_i^j$ is the row vector obtained from the $i^{\text{th}}$ row of $A$ by deleting the $j^{\text{th}}$ entry, and $c_j^i$ is the column vector obtained from the $j^{\text{th}}$ column of $A$ by deleting the $i^{\text{th}}$ entry. If the entries in $A$ commute, then
\begin{equation} \nonumber
    \left| A \right|_{ij} = (-1)^{i+j} \frac{\det(A)}{\det(A^{i,j})}.
\end{equation}
The derivative of a quasi-determinant follows two distinct formulas, depending on the structure of the underlying matrix. For a general quasi-determinant, the derivative takes the form
\begin{equation} \nonumber
    \begin{vmatrix}
        A & B \\
        C & \boxed{d}
    \end{vmatrix}^{\prime} = d' - C'A^{-1}B + CA^{-1}A'A^{-1}B - CA^{-1}B',
\end{equation}
where $A$ is an $n \times n$ matrix, $C$ is a row vector and $B$ a column vector. If $A$ is a Grammian-like matrix with derivative
\begin{equation} \nonumber
    A^{\prime} = \sum_{i=1}^k E_i F_i
\end{equation}
for some integer $k$, where $E_i$ ($F_i$) are column (row) vectors of appropriate length, then the derivative becomes
\begin{equation} \label{dquasigrammian}
    \begin{vmatrix}
        A & B \\
        C & \boxed{d}
    \end{vmatrix}^{\prime} = 
    \begin{vmatrix} 
        A & B \\ 
        C^{\prime} & \boxed{d^{\prime}} 
    \end{vmatrix} + 
    \begin{vmatrix} 
        A & B^{\prime} \\ 
        C & \boxed{0} 
    \end{vmatrix} + \sum_{i=1}^{k} 
    \begin{vmatrix} 
        A & E_i \\ 
        C & \boxed{0} 
    \end{vmatrix} 
    \begin{vmatrix} 
        A & B \\ 
        F_i & \boxed{0} 
    \end{vmatrix}.
\end{equation}
If the matrix $A$ does not have Grammian-like structure, then the derivative becomes
\begin{equation} \label{dquasiwronskian}
    \begin{split}
    \begin{vmatrix}
        A & B \\
        C & \boxed{d}
    \end{vmatrix}^{\prime} 
    & = 
    \begin{vmatrix}
        A & B \\
        C^{\prime} & \boxed{d^{\prime}}
    \end{vmatrix} 
        + \sum_{k=1}^{n} \begin{vmatrix}
        A & e_k \\
        C & \boxed{0}
    \end{vmatrix} 
    \begin{vmatrix}
        A & B \\
        (A^k)^{\prime} & \boxed{(B^k)^{\prime}}
    \end{vmatrix} \\
    & =
    \begin{vmatrix}
        A & B^{\prime} \\
        C & \boxed{d^{\prime}}
    \end{vmatrix} 
        + \sum_{k=1}^{n} 
    \begin{vmatrix}
        A & (A_k)^{\prime} \\
        C & \boxed{(C_k)^{\prime}}
    \end{vmatrix} 
    \begin{vmatrix}
        A & B \\
        e_k^\top & \boxed{0}
    \end{vmatrix},
\end{split}
\end{equation}
where $A^k$ ($A_k$) denotes the $k^\text{th}$ row (column) of the matrix $A$.

\section{Derivative Formulas}
\noindent In the appendix, we present explicit derivatives of $v = Q_n^{0, 0}$ with respect to $t_1, t_2, t_3, t_4$. The derivatives are arranged in ascending order according to their total differential order, where one differentiation with respect to 
$t_1$ is counted as order $1$, with respect to $t_2$ as order $2$, and so on.  The derivative formula could be obtained directly from \eqref{deri_quasi}.\\
$1^{st}$ order
\[
    v_{t_1} = Q_n^{1, 0} - Q_n^{0, 1} + Q_n^{0, 0}Q_n^{0, 0}
\] 
$2^{nd}$ order
\begin{align*}
   & v_{t_2} = Q_n^{2, 0} - Q_n^{0, 2} + Q_n^{0, 0}Q_n^{1, 0} + Q_n^{0, 1}Q_n^{0, 0}\\
    & v_{t_1t_1} = Q_n^{2, 0} - 2Q^{1, 1} + Q_n^{0, 2} + Q_n^{0, 0}[Q_n^{1, 0} - 2Q^{0, 1}] + [2Q^{1, 0} - Q_n^{0, 1}]Q_n^{0, 0} + 2(Q_n^{0, 0})^3
\end{align*}
$3^{rd}$ order
\begin{align*}
   & v_{t_3} = Q_n^{3, 0} - Q_n^{0, 3} + Q_n^{0, 0}Q_n^{2, 0} + Q_n^{0, 1}Q_n^{1, 0} + Q_n^{0, 2}Q_n^{0, 0}\\
    &v_{t_1t_2} = Q_n^{3, 0} - Q_n^{2, 1} - Q_n^{1, 2} + Q_n^{0, 3} + Q_n^{0, 0}[Q_n^{2, 0} - Q_n^{1, 1} - Q_n^{0, 2}] + [Q_n^{2, 0} + Q_n^{1, 1} - Q_n^{0, 2}]Q_n^{0, 0} \\ 
    & \qquad + Q_n^{1, 0}Q_n^{1, 0} - Q_n^{0, 1}Q_n^{0, 1} + Q_n^{0, 0}[Q_n^{1, 0} + Q_n^{0, 1}]Q_n^{0, 0} + (Q_n^{0, 0})^2Q^{1, 0} + Q_n^{0, 1}(Q_n^{0, 0})^2\\
    &v_{t_1t_1t_1} = Q_n^{3, 0} - 3Q^{2, 1} + 3Q^{1, 2} - Q_n^{0, 3} + 3Q^{1, 0}[Q_n^{1, 0} - 2Q^{0, 1}] - Q_n^{0, 1}[2Q^{1, 0} - 3Q^{0, 1}] \\ \nonumber
    & \qquad + Q_n^{0, 0}[Q_n^{2, 0} - 3Q^{1, 1} + 3Q^{0, 2}] + [3Q^{2, 0} - 3Q^{1, 1} + Q_n^{0, 2}]Q_n^{0, 0} + 3Q^{0, 0}[Q_n^{1, 0} - Q_n^{0, 1}]Q_n^{0, 0} \\ \nonumber
    & \qquad + 3(Q_n^{0, 0})^2[Q_n^{1, 0} - 2Q^{0, 1}] + 3[2Q^{1, 0} - Q_n^{0, 1}](Q_n^{0, 0})^2 + 6(Q_n^{0, 0})^4
\end{align*}
$4^{th}$ order
\begin{align*}
    &v_{t_4} = Q_n^{4, 0} - Q_n^{0, 4} + Q_n^{0, 0}Q_n^{3, 0} + Q_n^{0, 3}Q_n^{0, 0} + Q_n^{0, 1}Q_n^{2, 0} + Q_n^{0, 2}Q_n^{1, 0}\\
    &v_{t_2t_2} = Q_n^{4, 0} - 2Q^{2, 2} + Q_n^{0, 4} + Q_n^{0, 0}[Q_n^{3, 0} - 2Q^{1, 2}] + [2Q^{2, 1} - Q_n^{0, 3}]Q_n^{0, 0} + Q_n^{0, 1}[Q_n^{2, 0} - 2Q^{0, 2}] \\
    & \qquad + [2Q^{2, 0} - Q_n^{0, 2}]Q_n^{1, 0} + 2[Q_n^{0, 0}Q_n^{1, 0}Q_n^{1, 0} + Q_n^{0, 1}Q_n^{0, 1}Q_n^{0, 0} + Q_n^{0, 0}Q_n^{1, 1}Q_n^{0, 0} + Q_n^{0, 1}Q_n^{0, 0}Q_n^{1, 0}] \\
    &v_{t_1t_3} = Q_n^{4, 0} - Q_n^{3, 1} - Q_n^{1, 3} + Q_n^{0, 4} + Q_n^{0, 0}[Q_n^{3, 0} - Q_n^{2, 1} - Q_n^{0, 3}] + [Q_n^{3, 0} + Q_n^{1, 2} - Q_n^{0, 3}]Q_n^{0, 0} \\
    & \qquad - Q_n^{0, 2}Q_n^{0, 1} + Q_n^{1, 0}Q_n^{2, 0} + Q_n^{1, 1}Q_n^{1, 0} - Q_n^{0, 1}Q_n^{1, 1} + Q_n^{0, 0}[Q_n^{0, 2} + Q_n^{2, 0}]Q_n^{0, 0} + Q_n^{0, 0}Q_n^{0, 0}Q_n^{2, 0} \\ 
    & \qquad + Q_n^{0, 2}Q_n^{0, 0}Q_n^{0, 0} + Q_n^{0, 0}Q_n^{0, 1}Q_n^{1, 0} + Q_n^{0, 1}Q_n^{1, 0}Q_n^{0, 0} \\
\end{align*}
\begin{align*}
    &v_{t_1t_1t_1t_1} = Q_n^{4, 0} - 4Q^{3, 1} + 6Q^{2, 2} -4Q^{1, 3} + Q_n^{0, 4} + Q_n^{0, 0}[Q_n^{3, 0} - 4Q^{2, 1} + 6Q^{1, 2} - 4Q^{0, 3}] \\ 
    & \qquad + [4Q^{3, 0} - 6Q^{2, 1} + 4Q^{1, 2} - Q_n^{0, 3}]Q_n^{0, 0} - Q_n^{0, 1}[3Q^{2, 0} - 8Q^{1, 1} + 6Q^{0, 2}] + [6Q^{2, 0} \\ 
    & \qquad - 8Q^{1, 1} + 3Q^{0, 2}]Q_n^{1, 0} + 4Q^{1, 0}[Q_n^{2, 0} - 3Q^{1, 1} + 3Q^{0, 2}] - 4[3Q^{2, 0} - 3Q^{1, 1} + Q_n^{0, 2}]Q_n^{0, 1} \\ 
    & \qquad + 4 Q_n^{0, 0}Q_n^{0, 0}[Q_n^{2, 0} - 3Q^{1, 1} + 3Q^{0, 2}] + 2Q^{0, 0}[2Q^{2, 0} - 3Q^{1, 1} + 2Q^{0, 2}]Q_n^{0, 0} + 4[3Q^{2, 0} \\
    & \qquad - 3Q^{1, 1} + Q_n^{0, 2}]Q_n^{0, 0}Q_n^{0, 0} -4Q^{0, 0}Q_n^{0, 1}[2Q^{1, 0} - 3Q^{0, 1}] + 4[3Q^{1, 0} - 2Q^{0, 1}]Q_n^{1, 0}Q_n^{0, 0} \\
    & \qquad -6[2Q^{1, 0} - Q_n^{0, 1}]Q_n^{0, 1}Q_n^{0, 0} + 6[2Q^{1, 0}Q_n^{0, 0} - Q_n^{0, 1}Q_n^{0, 0} + Q_n^{0, 0}Q_n^{1, 0}][Q_n^{1, 0} - 2Q^{0, 1}] \\
    & \qquad + 12Q^{0, 0}Q_n^{0, 0}Q_n^{0, 0}[Q_n^{1, 0} - 2Q^{0, 1}] + 12[2Q^{1, 0} - Q_n^{0, 1}]Q_n^{0, 0}Q_n^{0, 0}Q_n^{0, 0} + 12Q^{0, 0}Q_n^{0, 0}[Q_n^{1, 0} \\ 
    & \qquad - Q_n^{0, 1}]Q_n^{0, 0} + 12Q^{0, 0}[Q_n^{1, 0} - Q_n^{0, 1}]Q_n^{0, 0}Q_n^{0, 0} + 24Q^{0, 0}Q_n^{0, 0}Q_n^{0, 0}Q_n^{0, 0}Q_n^{0, 0} 
\end{align*}
$5^{th}$ order
\begin{align*}
    &v_{t_1t_4} = Q_n^{5, 0} - Q_n^{4, 1} - Q_n^{1, 4} + Q_n^{0, 5} + Q_n^{0, 0}[Q_n^{4, 0} - Q_n^{3, 1} - Q_n^{0, 4}] + [Q_n^{4, 0} + Q_n^{1, 3} - Q_n^{0, 4}]Q_n^{0, 0} \\
    & \qquad + Q_n^{1, 2}Q_n^{1, 0} - Q_n^{0, 1}Q_n^{2, 1} + Q_n^{1, 1}Q_n^{2, 0} - Q_n^{0, 2}Q_n^{1, 1} + Q_n^{1, 0}Q_n^{3, 0} - Q_n^{0, 3}Q_n^{0, 1} + Q_n^{0, 0}[Q_n^{3, 0} \\ 
    & \qquad + Q_n^{0, 3}]Q_n^{0, 0} + \{Q_n^{0, 0}, Q_n^{0, 1}Q_n^{2, 0} + Q_n^{0, 2}Q_n^{1, 0}\} + Q_n^{0, 0}Q_n^{0, 0}Q_n^{3, 0} + Q_n^{0, 3}Q_n^{0, 0}Q_n^{0, 0}\\
    &v_{t_2t_3} = Q_n^{5, 0} - Q_n^{3, 2} - Q_n^{2, 3} + Q_n^{0, 5} + Q_n^{0, 0}[Q_n^{4, 0} - Q_n^{2, 2} - Q_n^{1, 3}] + [Q_n^{3, 1} + Q_n^{2, 2} - Q_n^{0, 4}]Q_n^{0, 0} \\
    & \qquad + Q_n^{0, 1}[Q_n^{3, 0} - Q_n^{1, 2} - Q_n^{0, 3}] + [Q_n^{3, 0} + Q_n^{2, 1} - Q_n^{0, 3}]Q_n^{1, 0} + Q_n^{2, 0}Q_n^{2, 0} - Q_n^{0, 2}Q_n^{0, 2} \\ 
    & \qquad + Q_n^{0, 0}[Q_n^{2, 1} + Q_n^{1, 2}]Q_n^{0, 0} + Q_n^{0, 0}[Q_n^{1, 1}Q_n^{1, 0} + \{Q_n^{1, 0}, Q_n^{2, 0}\}] + [Q_n^{0, 1}Q_n^{1, 1} + \{Q_n^{0, 1}, Q_n^{0, 2}\}]Q_n^{0, 0} \\ 
    & \qquad + Q_n^{0, 1}[Q_n^{1, 0}Q_n^{1, 0} + Q_n^{0, 0}Q_n^{2, 0}] + [Q_n^{0, 1}Q_n^{0, 1} + Q_n^{0, 2}Q_n^{0, 0}]Q_n^{1, 0}
\end{align*}
\begin{align*}
    &v_{t_1t_1t_1t_2} = Q_n^{5, 0} - 3Q^{4, 1} + 2Q^{3, 2} + 2Q^{2, 3} - 3Q^{1, 4} + Q_n^{0, 5} + Q_n^{0, 0}[Q_n^{4, 0} - 3Q^{3, 1} + 2Q^{2, 2} + 2Q^{1, 3} \\
    & \qquad - 3Q^{0, 4}] + [3Q^{4, 0} - 2Q^{3, 1} - 2Q^{2, 2} + 3Q^{1, 3} - Q_n^{0, 4}]Q_n^{0, 0} - Q_n^{0, 1}[2Q^{3, 0} - 3Q^{2, 1} - 2Q^{1, 2} \\
    & \qquad + 4Q^{0, 3}] + [4Q^{3, 0} - 2Q^{2, 1} - 3Q^{1, 2} + 2Q^{0, 3}]Q_n^{1, 0} + 3Q^{1, 0}[Q_n^{3, 0} - 2Q^{2, 1} + 2Q^{0, 3}] \\ 
    & \qquad - 3[2Q^{3, 0} - 2Q^{1, 2} + Q_n^{0, 3}]Q_n^{0, 1} - 3Q^{1, 1}[Q_n^{2, 0} - 2Q^{0, 2}] - 3[2Q^{2, 0} - Q_n^{0, 2}]Q_n^{1, 1} \\
    & \qquad + 4[Q_n^{2, 0}Q_n^{2, 0} - Q_n^{0, 2}Q_n^{0, 2}] + 3Q^{0, 0}Q_n^{0, 0}[Q_n^{3, 0} - 2Q^{2, 1} + 2Q^{0, 3}] + 3[2Q^{3, 0} - 2Q^{1, 2} \\ 
    & \qquad + Q_n^{0, 3}]Q_n^{0, 0}Q_n^{0, 0} + Q_n^{0, 0}[3Q^{3, 0} - 2Q^{2, 1} - 2Q^{1, 2} + 3Q^{0, 3}]Q_n^{0, 0} + Q_n^{0, 0}[6\{Q_n^{0, 1}, Q_n^{0, 2}\} \\ 
    & \qquad + 4\{Q_n^{1, 0}, Q_n^{2, 0}\} - 6Q^{2, 0}Q_n^{0, 1} - 3Q^{0, 1}Q_n^{2, 0} - 3Q^{0, 2}Q_n^{1, 0} - 6Q^{1, 0}Q_n^{1, 1} - 2Q^{1, 1}Q_n^{1, 0}] \\
    & \qquad + [6\{Q_n^{1, 0}, Q_n^{2, 0}\} + 4\{Q_n^{0, 1}, Q_n^{0, 2}\} - 6Q^{1, 0}Q_n^{0, 2} - 3Q^{0, 2}Q_n^{1, 0} - 3Q^{0, 1}Q_n^{2, 0} - 6Q^{1, 1}Q_n^{0, 1} \\
    & \qquad - 2Q^{0, 1}Q_n^{1, 1}]Q_n^{0, 0} + 6Q^{1, 0}Q_n^{0, 0}[Q_n^{2, 0} - Q_n^{1, 1} - Q_n^{0, 2}] - 6[Q_n^{2, 0} + Q_n^{1, 1} - Q_n^{0, 2}]Q_n^{0, 0}Q_n^{0, 1} \\ 
    & \qquad + 2[3Q^{2, 0} - Q_n^{0, 2}]Q_n^{0, 0}Q_n^{1, 0} - 2Q^{0, 1}Q_n^{0, 0}[Q_n^{2, 0} - 3Q^{0, 2}] + 6Q^{1, 0}Q_n^{1, 0}Q_n^{1, 0} + 6Q^{0, 1}Q_n^{0, 1}Q_n^{0, 1} \\ 
    & \qquad - 6Q^{1, 0}Q_n^{1, 0}Q_n^{0, 1} - 2Q^{0, 1}Q_n^{1, 0}Q_n^{1, 0} - 6Q^{1, 0}Q_n^{0, 1}Q_n^{0, 1} - 2Q^{0, 1}Q_n^{0, 1}Q_n^{1, 0} + 6[Q_n^{2, 0} + Q_n^{1, 1} \\ 
    & \qquad - Q_n^{0, 2}]Q_n^{0, 0}Q_n^{0, 0}Q_n^{0, 0} + 6Q^{0, 0}Q_n^{0, 0}Q_n^{0, 0}[Q_n^{2, 0} - Q_n^{1, 1} - Q_n^{0, 2}] + 6Q^{0, 0}Q_n^{0, 0}[Q_n^{2, 0} - Q_n^{0, 2}]Q_n^{0, 0} \\ 
    & \qquad + 6Q^{0, 0}[Q_n^{2, 0} - Q_n^{0, 2}]Q_n^{0, 0}Q_n^{0, 0} + 6\{Q_n^{0, 0}Q_n^{0, 0}, Q_n^{1, 0}Q_n^{1, 0} - Q_n^{0, 1}Q_n^{0, 1}\} + 6Q^{0, 0}[Q_n^{1, 0}Q_n^{1, 0} \\ 
    & \qquad - Q_n^{0, 1}Q_n^{0, 1}]Q_n^{0, 0} + 6Q^{0, 0}[Q_n^{1, 0}Q_n^{0, 0}Q_n^{1, 0} - Q_n^{1, 0}Q_n^{0, 0}Q_n^{0, 1} - Q_n^{0, 1}Q_n^{0, 0}Q_n^{0, 1}] \\ 
    & \qquad + 6[Q_n^{1, 0}Q_n^{0, 0}Q_n^{1, 0} + Q_n^{1, 0}Q_n^{0, 0}Q_n^{0, 1} - Q_n^{0, 1}Q_n^{0, 0}Q_n^{0, 1}]Q_n^{0, 0} + 6[Q_n^{1, 0}Q_n^{0, 1}, Q_n^{0, 0}Q_n^{0, 0}] \\ 
    & \qquad + 6Q^{1, 0}Q_n^{0, 0}Q_n^{0, 0}Q_n^{1, 0} - 6Q^{0, 1}Q_n^{0, 0}Q_n^{0, 0}Q_n^{0, 1} + 6Q^{0, 0}Q_n^{0, 0}Q_n^{0, 0}[Q_n^{1, 0} + Q_n^{0, 1}]Q_n^{0, 0} \\
    & \qquad + 6Q^{0, 0}[Q_n^{1, 0} + Q_n^{0, 1}]Q_n^{0, 0}Q_n^{0, 0}Q_n^{0, 0} + 6Q^{0, 0}Q_n^{0, 0}[Q_n^{1, 0} + Q_n^{0, 1}]Q_n^{0, 0}Q_n^{0, 0} \\ 
    & \qquad + 6Q^{0, 0}Q_n^{0, 0}Q_n^{0, 0}Q_n^{0, 0}Q_n^{1, 0} + 6Q^{0, 1}Q_n^{0, 0}Q_n^{0, 0}Q_n^{0, 0}Q_n^{0, 0}    
\end{align*}

It is worth noting that we have not presented the full set for 4- and 5-order derivatives of $v$, for example, we missed the expressions for $v_{t_1t_1t_2}$, $v_{t_1t_2t_2}$, $v_{t_1t_1t_3}$ or $v_{t_1t_1t_1t_1t_1}$, since these derivatives are not involved in the ncKP or ncDJKM equations presented in this work.


\begin{thebibliography}{sl}

\bibitem{asakawa00} 
T. Asakawa and I. Kishimoto. 
Noncommutative gauge theories from deformation quantization. 
\emph{Nucl. Phys. B}, 591 (2000), 611-635.

\bibitem{bobrova24}
I. Bobrova, V. Retakh, V. Rubstov and G. Sharygin. 
Non-abelian discrete Toda chains and related lattices.
\emph{Phys. D}, 464 (2024), 134200.

\bibitem{castellani00} 
L. Castellani.
Non-commutative geometry and physics: a review of selected recent results. 
\emph{Classical Quant. Grav.}, 17 (2000), 3377.

\bibitem{date82} 
E. Date, M. Jimbo and M. Kashiwara.
Transformation groups for soliton equations: IV. A new hierarchy of soliton equations of KP-type.
\emph{Physica D}, 4 (1982), 343-365.

\bibitem{doliwa23}
A. Doliwa and A. Siemaszko.
Integrability and geometry of the Wynn recurrence.
\emph{Numer. Algo.}, 92 (2023), 571-596.

\bibitem{etingof97} 
P. Etingof, I. Gelfand and V. Retakh.
Factorization of differential operators, quasideterminants, and nonabelian Toda field equations. 
\emph{Math. Res. Lett.}, 4 (1997), 413-425.

\bibitem{feng22}
B. Feng and Y. Ohta.
General rogue wave solution to the discrete nonlinear Schrödinger equation.
\emph{Phys. D}, 439 (2022), 133400.

\bibitem{gelfand91} 
I. Gelfand and S. Retakh. 
Determinants of matrices over noncommutative rings. 
\emph{Funct. Anal. Its Appl.}, 25 (1991), 91-102. 

\bibitem{gelfand05} 
I. Gelfand, S. Gelfand and V. Retakh. 
Quasideterminants. 
\emph{Adv. Math.}, 193 (2005), 56-141.

\bibitem{gilson07} 
C. Gilson and J. Nimmo. 
On a direct approach to quasideterminant solutions of a noncommutative KP equation. \emph{J. Phys. A: Math. Theor.}, 40 (2007), 3839. 

\bibitem{gilson08}
C. Gilson, M. Hamanaka and J. Nimmo.
B\"acklund transformations for non-commutative anti-self-dual Yang-Mills equations.
\emph{Glasgow Math. J.}, 51A (2009), 83-93.

\bibitem{gilson09} 
C. Gilson, J. Nimmo and C. Sooman. 
Matrix solutions of a noncommutative KP equation and a noncommutative mKP equation. \emph{Theor. Math. Phys.}, 159 (2009), 796-805.

\bibitem{gilson25}
C. Gilson, S. Li and Y. Shi.
Matrix-valued theta-deformed bi-orthogonal polynomials, non-commutative Toda theory and B\"acklund transformation. \emph{Nonlinearity}, 38 (2025), 065003.

\bibitem{hirota71} 
R. Hirota. 
Exact solution of the Korteweg—de Vries equation for multiple collisions of solitons. \emph{Phys. Rev. Lett.}, 27 (1971), 1192.

\bibitem{hamanaka03} 
M. Hamanaka and K. Toda. 
Towards noncommutative integrable systems. 
\emph{Phys. Lett. A}, 316 (2003), 77-83. 

\bibitem{hirota04} 
R. Hirota. 
\emph{The direct Method in Soliton Theory.} 
(Cambridge University Press, Cambridge, 2004).

\bibitem{hamanaka05} 
M. Hamanaka. 
Commuting flows and conservation laws for noncommutative Lax hierarchies. 
\emph{J. Math. Phys.}, 46 (2005), 052701.

\bibitem{hamanaka07} 
M. Hamanaka. 
Notes on exact multi-soliton solutions of noncommutative integrable hierarchies. 
\emph{J. High Energy Phys.}, 2007 (2007), 094. 

\bibitem{hamanaka10} 
M. Hamanaka. 
Noncommutative integrable systems and quasideterminants. 
\emph{AIP Conf. Proc.}, 1212 (2010), 122-135.

\bibitem{kodama24} 
S. Huang, Y. Kodama and C. Li. 
Non-crossing permutations for the KP solitons under the Gel'fand-Dickey reductions and the vertex operators.
\emph{arXiv preprint arXiv:2407.01900} (2024).


\bibitem{jimbo83} 
M. Jimbo and T. Miwa.
Solitons and infinite dimensional Lie algebras.
\emph{Publ. RIMS, Kyoto Univ.}, 19 (1983), 943-1001.

\bibitem{landi03} 
G. Landi. 
\emph{An introduction to noncommutative spaces and their geometries.} 
(Springer Science \& Business Media, 2003).

\bibitem{li07} 
C. Li, W. Ma and X. Liu. 
Wronskian solutions of the Boussinesq equation-solitons, negatons, positons and complexitons. 
\emph{Inverse Probl.}, 23 (2007), 279.

\bibitem{li24}
S. Li.
Matrix Orthogonal Polynomials, non-abelian Toda lattice and Bäcklund transformation.
\emph{Sci. China Math.}, 67 (2024), 2071-2090.

\bibitem{ma05} 
W. Ma and Y. You. 
Solving the Korteweg-de Vries equation by its bilinear form: Wronskian solutions.
\emph{Trans. Am. Math. Soc.}, 357 (2005), 1753-1778.

\bibitem{nimmo83}
J. Nimmo. 
A bilinear B\"acklund transformation for the nonlinear Schr\"odinger equation.
\emph{Phys. Lett. A}, 99 (1983), 279-280.

\bibitem{nimmo84}
J. Nimmo and N. Freeman.
The use of Backlund transformations in obtaining N-soliton solutions in Wronskian form.
\emph{J. Phys. A}, 17 (1984), 1415.

\bibitem{ohta14}
Y. Ohta and J. Yang.
General rogue waves in the focusing and defocusing Ablowitz-Ladik equations.
\emph{J. Phys. A}, 47 (2014), 255201.

\bibitem{paniak01} 
L. Paniak. 
Exact noncommutative KP and KdV multi-solitons. 
\emph{arXiv preprint hep-th/0105185} (2001).

\bibitem{sato83} 
M. Sato. 
Soliton equations as dynamical systems on infinite dimensional Grassmann manifold.
\emph{North-Holland Mathematics Studies, North-Holland, (1983)}, pp. 259-271.

\bibitem{szabo03} 
R. Szabo. 
Quantum field theory on noncommutative spaces.
\emph{Phys. Rept.}, 378 (2003), 207-299.

\bibitem{sakakibara04} 
M. Sakakibara. 
Factorization methods for noncommutative KP and Toda hierarchy. 
\emph{J. Phys A: Math. Gen.}, 37 (2004), L599.

\bibitem{wang25}
B. Wang and S. Li.
On non-commutative leapfrog map. 
\emph{J. Lond. Math. Soc.}, 111 (2025), e70063.

\end{thebibliography}
\end{document}